%% file: paper.tex
\documentclass{llncs}
\usepackage[utf8]{inputenc}
\usepackage{amsmath,amssymb,mathpazo}

\usepackage[colorlinks=true,allcolors=blue,breaklinks,draft=false]{hyperref} 
\usepackage[rightcaption]{sidecap}
\usepackage{booktabs}
\usepackage[all]{xy}

\usepackage{etex}
\usepackage{mathrsfs}
\usepackage{mathtools}
\usepackage{stmaryrd}

\include{macros}

\title{On partial order semantics for SAT/SMT-based symbolic encodings of weak memory concurrency\thanks{This work is funded by a gift from Intel Corporation for research on Effective Validation of Firmware and the ERC project ERC 280053.}}
\author{Alex Horn \and Daniel Kroening}

\institute{University of Oxford}

\begin{document}

\maketitle
\begin{abstract}
Concurrent systems are notoriously difficult to analyze, and technological advances such as weak memory architectures greatly compound this problem.
This has renewed interest in partial order semantics as a theoretical foundation for formal verification techniques.
Among these, symbolic techniques have been shown to be particularly effective at finding concurrency-related bugs because they can leverage highly optimized decision procedures such as SAT/SMT solvers.
This paper gives new fundamental results on partial order semantics for SAT/SMT-based symbolic encodings of weak memory concurrency.
In particular, we give the theoretical basis for a decision procedure that can handle a fragment of concurrent programs endowed with least fixed point operators.
In addition, we show that a certain partial order semantics of relaxed sequential consistency is equivalent to the conjunction of three extensively studied weak memory axioms by Alglave et al.
An important consequence of this equivalence is an asymptotically smaller symbolic encoding for bounded model checking which has only a quadratic number of partial order constraints compared to the state-of-the-art cubic-size encoding.
\end{abstract}

\section{Introduction}
\vspace{-0.3em}
\label{section:introduction}

Concurrent systems are notoriously difficult to analyze, and technological advances such as weak memory architectures as well as highly available distributed services greatly compound this problem. This has renewed interest in partial order concurrency semantics as a theoretical foundation for formal verification techniques. Among these, \emph{symbolic techniques} have been shown to be particularly effective at finding concurrency-related bugs because they can leverage highly optimized decision procedures such as SAT/SMT solvers. This paper studies partial order semantics from the perspective of SAT/SMT-based symbolic encodings of weak memory concurrency.

Given the diverse range of partial order concurrency semantics, we link our study to a recently developed unifying theory of concurrency by Tony Hoare et al.~\cite{HMSW2011}. This theory is known as \emph{Concurrent Kleene Algebra} (CKA) which is an algebraic concurrency semantics based on quantales, a special case of the fundamental algebraic structure of idempotent semirings. Based on quantales, CKA combines the familiar laws of the sequential program operator ($;$) with a new operator for concurrent program composition ($\parallel$). A distinguishing feature of CKA is its exchange law $(\cU \parallel \cV) ; (\cX \parallel \cY) \subseteq (\cU ; \cX) \parallel (\cV ; \cY)$ that describes how sequential and concurrent composition operators can be interchanged. Intuitively, since the binary relation $\subseteq$ denotes program refinement, the exchange law expresses a divide-and-conquer mechanism for how concurrency may be sequentially implemented on a machine. The exchange law, together with a uniform treatment of programs and their specifications, is key to unifying existing theories of concurrency~\cite{HvS2014}. CKA provides such a unifying theory~\cite{HvS2012,HvS2014} that has practical relevance on proving program correctness, e.g. using rely/guarantee reasoning~\cite{HMSW2011}. Conversely, however, pure algebra cannot refute that a program is correct or that certain properties about every program always hold~\cite{HvS2012,HvS2014,HvSMSVZOH2014}. This is problematic for theoretical reasons but also in practice because todays software complexity requires a diverse set of program analysis tools that range from proof assistants to automated testing. The solution is to accompany CKA with a mathematical model which satisfies its laws so that we can \emph{prove} as well as \emph{disprove} properties about programs.

One such well-known model-theoretical foundation for CKA is Pratt's~\cite{P1986} and Gischer's~\cite{G1988} partial order model of computation that is constructed from \emph{labelled partially ordered multisets} (pomsets). Pomsets generalize the concept of a string in finite automata theory by relaxing the total ordering of the occurrence of letters within a string to a partial order. For example, $a \parallel a$ denotes a pomset that consists of two unordered events that are both labelled with the letter $a$. By partially ordering events, pomsets form an integral part of the extensive theoretical literature on so-called `true concurrency', e.g.~\cite{P1966,L1978,G1981,NPW1981,P1986,G1988}, in which pomsets strictly generalize Mazurkiewicz traces~\cite{BK1992}, and prime event structures~\cite{NPW1981} are pomsets enriched with a conflict relation subject to certain conditions. From an algorithmic point of view, the complexity of the \emph{pomset language membership} (PLM) problem is NP-complete, whereas the pomset language containment (PLC) problem is $\Pi_2^p$-complete~\cite{FKL1993}.

Importantly, these aforementioned theoretical results only apply to star-free pomset languages (without fixed point operators). In fact, the decidability of the equational theory of the pomset language closed under least fixed point, sequential and concurrent composition operators (but without the exchange law) has been only most recently established~\cite{LS2014}; its complexity remains an open problem~\cite{LS2014}. Yet another open problem is the decidability of this equational theory together with the exchange law~\cite{LS2014}. In addition, it is still unclear how theoretical results about pomsets may be applicable to formal techniques for finding concurrency-related bugs. In fact, it is not even clear how insights about pomsets may be combined with most recently studied language-specific or hardware-specific concurrency semantics, e.g.~\cite{SSONM2010,SVZJS2011,BOSSW2011,AMSS2012}.

These gaps are motivation to reinvestigate pomsets from an algorithmic perspective. In particular, our work connects pomsets to a SAT/SMT-based bounded model checking technique~\cite{AKT2013} where shared memory concurrency is symbolically encoded as partial orders. To make this connection, we adopt pomsets as \emph{partial strings} (Definition~\ref{def:partial-string}) that are ordered by a refinement relation (Definition~\ref{def:partial-string-isomorphism}) based on \'{E}sik's notion of \emph{monotonic bijective morphisms}~\cite{E2002}. Our partial-string model then follows from the standard Hoare powerdomain construction where sets of partial strings are downward-closed with respect to monotonic bijective morphism (Definition \ref{def:program}). The relevance of this formalization for the modelling of weak memory concurrency (including data races) is explained through several examples. Our main contributions are as follows:
\begin{enumerate}
\item We give the theoretical basis for a decision procedure that can handle a fragment of \emph{concurrent programs endowed with least fixed point operators} (Theorem~\ref{theorem:program-reduction}). This is accomplished by exploiting a form of periodicity, thereby giving a mechanism for reducing a countably infinite number of events to a finite number. This result particularly caters to partial order encoding techniques that can currently only encode a finite number of events due to the deliberate restriction to quantifier-free first-order logic, e.g.~\cite{AKT2013}.
\item We then interpret a particular form of weak memory in terms of certain downward-closed sets of partial strings (Definition~\ref{def:SC-relaxed-program}), and show that our interpretation is equivalent to the conjunction of three fundamental weak memory axioms (Theorem~\ref{theorem:SC-relaxed-equivalence}), namely `write coherence', `from-read' and `global read-from'~\cite{AMSS2012}. Since all three axioms underpin extensive experimental research into weak memory architectures~\cite{AMSS2011}, \emph{Theorem~\ref{theorem:SC-relaxed-equivalence} gives denotational partial order semantics a new practical dimension}.
\item Finally, we prove that there exists an \emph{asymptotically smaller quantifier-free first-order logic formula} that has only $O(N^2)$ partial order constraints (Theorem~\ref{theorem:smaller-fr}) compared to the state-of-the-art $O(N^3)$ partial order encoding for bounded model checking~\cite{AKT2013} where $N$ is the maximal number of reads and writes on the same shared memory address. This is significant because $N$ can be prohibitively large when concurrent programs frequently share data.
\end{enumerate}

The rest of this paper is organized into three parts. First, we recall familiar concepts on partial-string theory (\autoref{section:partial-string-theory}) on which the rest of this paper is based. We then prove a least fixed point reduction result (\autoref{label:least-fixed-point-reduction}). Finally, we characterize a particular form of relaxed sequential consistency in terms of three weak memory axioms by Alglave et al. (\autoref{section:SC-relaxed}).

\section{Partial-string theory}
\vspace{-0.3em}
\label{section:partial-string-theory}

In this section, we adapt an axiomatic model of computation that uses partial orders to describe the semantics of concurrent systems. For this, we recall familiar concepts (Definition~\ref{def:partial-string},~\ref{def:partial-string-composition},~\ref{def:partial-string-isomorphism}~and~\ref{def:program}) that underpin our mathematical model of CKA (Theorem~\ref{theorem:program-algebra}). This model is the basis for subsequent results in~\autoref{label:least-fixed-point-reduction}~and~\autoref{section:SC-relaxed}.

\begin{definition}[Partial string]
\label{def:partial-string}
Denote with $E$ a nonempty set of \defn{events}. Let $\Gamma$ be an \defn{alphabet}. A \defn{partial string} $p$ is a triple $\tuple{E_p, \alpha_p, \preceq_p}$ where $E_p$ is a subset of $E$, $\alpha_p \colon E_p \to \Gamma$ is a function that maps each event in $E_p$ to an alphabet symbol in $\Gamma$, and $\preceq_p$ is a partial order on $E_p$. Two partial strings $p$ and $q$ are said to be \defn{disjoint} whenever $E_p \cap E_q = \emptyset$. A partial string $p$ is called \defn{empty} whenever $E_p = \emptyset$. Denote with $\sP_f$ the set of all \defn{finite partial strings} $p$ whose event set $E_p$ is finite.
\end{definition}

Each event in the universe $E$ should be thought of as an occurrence of a computational step, whereas letters in $\Gamma$ describe the computational effect of events. Typically, we denote a partial string by $p$, or letters from $x$ through $z$. In essence, a partial string $p$ is a partially-ordered set $\pair{E_p}{\preceq_p}$ equipped with a labelling function $\alpha_p$. A partial string is therefore the same as a \emph{labelled partial order} (lpo), see also Remark~\ref{remark:partial-string}. We draw finite partial strings in $\sP_f$ as inverted Hasse diagrams (e.g. Fig.~\ref{fig:partial-string-example}), where the ordering between events may be interpreted as a happens-before relation~\cite{L1978}, a fundamental notion in distributed systems and formal verification of concurrent systems, e.g.~\cite{BOSSW2011,AMSS2012}. We remark the obvious fact that the empty partial string is unique under component-wise equality.

\begin{example}
In the partial string in Fig.~\ref{fig:partial-string-example}, $e_0$ happens-before $e_1$, whereas both $e_0$ and $e_2$ happen concurrently because neither $e_0 \preceq_p e_2$ nor $e_2 \preceq_p e_0$.
\end{example}

\begin{SCfigure}[100][t]
\xymatrix@C=1.2em@R=0.9em{
  e_0            & e_2            \\
  e_1\ar@{<-}[u] & e_3\ar@{<-}[u]
}
\caption{
A partial string $p = \tuple{E_p, \alpha_p, \preceq_p}$ with events $E_p = \set{e_0, e_1, e_2, e_3}$ and the labelling function $\alpha_p$ satisfying the following: $\alpha_p(e_0) = `r_0\, \texttt{:=}\, [b]_\mathsf{acquire}\textrm'$, $\alpha_p(e_1) = `r_1\, \texttt{:=}\, [a]_\mathsf{none}\textrm'$, $\alpha_p(e_2) = `[a]_\mathsf{none}\,\texttt{:=}\,1\textrm'$ and $\alpha_p(e_3) = `[b]_\mathsf{release}\,\texttt{:=}\,1\textrm'$.
}
\label{fig:partial-string-example}
\end{SCfigure}

We abstractly describe the control flow in concurrent systems by adopting the sequential and concurrent operators on labelled partial orders~\cite{G1981,P1986,G1988,E2002,HA2014}.

\begin{definition}[Partial string operators]
\label{def:partial-string-composition}
Let $x$ and $y$ be disjoint partial strings. Let $x \parallel y \deq \tuple{E_{x \parallel y}, \alpha_{x \parallel y}, \preceq_{x \parallel y}}$ and $x ; y \deq \tuple{E_{x ; y}, \alpha_{x ; y}, \preceq_{x ; y}}$ be their \defn{concurrent} and \defn{sequential composition}, respectively, where $E_{x \parallel y} = E_{x ; y} \deq E_x \cup E_y$ such that, for all events $e, e'$ in $E_x \cup E_y$, the following holds:
\begin{itemize}
\item $e \preceq_{x \parallel y} e' \text{ exactly  if } e \preceq_x e' \text{ or } e \preceq_y e'$,
\item $e \preceq_{x ; y} e' \text{ exactly if } (e \in E_x \text{ and } e' \in E_y) \text{ or } e \preceq_{x \parallel y} e'$,
\item $\alpha_{x \parallel y}(e) = \alpha_{x ; y}(e) \deq
\begin{cases}
  \alpha_x(e) &\text{if } e \in E_x \\
  \alpha_y(e) &\text{if } e \in E_y.
\end{cases}$
\end{itemize}
\end{definition}

For simplicity, we assume that partial strings can be always made disjoint by renaming events if necessary. But this assumption could be avoided by using coproducts, a form of constructive disjoint union~\cite{HA2014}. When clear from the context, we construct partial strings directly from the labels in $\Gamma$.

\begin{example}
If we ignore labels for now and let $p_i$ for all $0 \leq i \leq 3$ be four partial strings which each consist of a single event $e_i$, then $(p_0 ; p_1) \parallel (p_2 ; p_3)$ corresponds to a partial string that is isomorphic to the one shown in Fig.~\ref{fig:partial-string-example}.
\end{example}

To formalize the set of all possible happens-before relations of a concurrent system, we rely on \'{E}sik's notion of monotonic bijective morphism~\cite{E2002}:

\begin{definition}[Partial string refinement]
\label{def:partial-string-isomorphism}
Let $x$ and $y$ be partial strings such that $x = \tuple{E_x, \alpha_x \preceq_x}$ and $y = \tuple{E_y, \alpha_y, \preceq_y}$. A \defn{monotonic bijective morphism} from $x$ to $y$, written $f \colon x \to y$, is a bijective function $f$ from $E_x$ to $E_y$ such that, for all events $e, e' \in E_x$, $\alpha_x(e) = \alpha_y(f(e))$, and if $e \preceq_x e'$, then $f(e) \preceq_y f(e')$. Then $x$ \defn{refines} $y$, written $x \sqsubseteq y$, if there exists a monotonic bijective morphism $f \colon y \to x$ from $y$ to $x$.
\end{definition}

\begin{remark}
\label{remark:partial-string}
Partial words~\cite{G1981} and pomsets~\cite{P1986,G1988} are defined in terms of isomorphism classes of lpos. Unlike lpos in pomsets, however, we study partial strings in terms of monotonic bijective morphisms~\cite{E2002} because isomorphisms are about sameness whereas the exchange law on partial strings is an inequation~\cite{HA2014}.
\end{remark}

\begin{SCfigure}[100][t]
\xy
  \xymatrix "M"@C=1.2em@R=1em{
    e_0\ar@{.>}@/^1pc/[rrr]             & e_2\ar@{.>}@/^1pc/[rrr]             &              & e'_0             & e'_2           \\
    e_1\ar@{<-}[u]\ar@{.>}@/^-1pc/[rrr] & e_3\ar@{<-}[u]\ar@{.>}@/^-1pc/[rrr] &   \ar@{~}[u] & e'_1\ar@{<-}[u] & e'_3\ar@{<-}[u]\ar@{<-}[lu] \\
&&&&\\
  }
\POS"M3,1"."M3,2"!C*\frm{_\}},+D*++!U\txt{$y$}
    ,"M3,4"."M3,5"!C*\frm{_\}},+D*++!U\txt{$x$}
\endxy
\caption{Two partial strings $x$ and $y$ such that $x \sqsubseteq y$ provided all the labels are preserved, e.g. $\alpha_x(e'_0) = \alpha_y(e_0)$.}
\label{fig:sqsubseteq-partial-string}
\end{SCfigure}
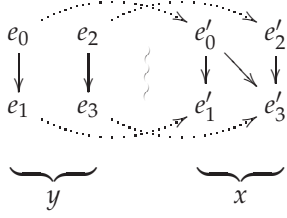

The purpose of Definition~\ref{def:partial-string-isomorphism} is to disregard the identity of events but retain the notion of `subsumption', cf.~\cite{G1988}. The intuition is that $\sqsubseteq$ orders partial strings according to their determinism. In other words, $x \sqsubseteq y$ for partial strings $x$ and $y$ implies that all events ordered in $y$ have the same order in $x$.

\begin{example}
\label{example:sqsubseteq-partial-string}
Fig.~\ref{fig:sqsubseteq-partial-string} shows a monotonic bijective morphism from a partial string as given in Fig.~\ref{fig:partial-string-example} to an $N$-shaped partial string that is almost identical to the one in Fig.~\ref{fig:partial-string-example} except that it has an additional partial order constraint, giving its $N$ shape. One well-known fact about $N$-shaped partial strings is that they cannot be constructed as $x ; y$ or $x \parallel y$ under any labelling~\cite{P1986}. However, this is not a problem for our study, as will become clear after Definition~\ref{def:program}.
\end{example}

Our notion of partial string refinement is particularly appealing for symbolic techniques of concurrency because the monotonic bijective morphism can be directly encoded as a first-order logic formula modulo the theory of uninterpreted functions. Such a symbolic partial order encoding would be fully justified from a computational complexity perspective, as shown next.

\begin{proposition}
\label{proposition:psr-NP-complete}
Let $x$ and $y$ be finite partial strings in $\sP_f$. The \defn{partial string refinement} (PSR) problem --- i.e. whether $x \sqsubseteq y$ --- is NP-complete.
\end{proposition}
\begin{proof}
Clearly \textsc{PSR} is in NP. 
The NP-hardness proof proceeds by reduction from the PLM problem~\cite{FKL1993}. Let $\Gamma^\ast$ be the set of strings, i.e. the set of finite partial strings $s$ such that $\preceq_s$ is a total order (for all $e, e' \in E_s$, $e \preceq_s e'$ or $e' \preceq_s e$). Given a finite partial string $p$, let $\mfrL_p$ be the set of all strings which refine $p$; equivalently, $\mfrL_p \deq \set{s \in \Gamma^\ast \alt s \sqsubseteq p}$. So $\mfrL_p$ denotes the same as $L(p)$ in~\cite[Definition 2.2]{FKL1993}.

Let $s$ be a string in $\Gamma^\ast$ and $P$ be a pomset over the alphabet $\Gamma$. By Remark~\ref{remark:partial-string}, fix $p$ to be a partial string in $P$. Thus $s$ refines $p$ if and only if $s$ is a member of $\mfrL_p$. Since this membership problem is NP-hard~\cite[Theorem 4.1]{FKL1993}, it follows that the \textsc{PSR} problem is NP-hard. So the \textsc{PSR} problem is NP-complete. \qed
\end{proof}

Note that a single partial string is not enough to model mutually exclusive (nondeterministic) control flow. To see this, consider a simple (possibly sequential) system such as \texttt{if * then P else Q} where \texttt{*} denotes nondeterministic choice. If the semantics of a program was a single partial string, then we need to find exactly one partial string that represents the fact that \texttt{P} executes or \texttt{Q} executes, but never both. To model this, rather than using a conflict relation~\cite{NPW1981}, we resort to the simpler Hoare powerdomain construction where we lift sequential and concurrent composition operators to \emph{sets} of partial strings. But since we are aiming (similar to Gischer~\cite{G1988}) at an \emph{over-approximation of concurrent systems}, these sets are downward closed with respect to our partial string refinement ordering from Definition~\ref{def:partial-string-isomorphism}. Additional benefits of using the downward closure include that program refinement then coincides with familiar set inclusion and the ease with which later the Kleene star operators can be defined.

\begin{definition}[Program]
\label{def:program}
A \defn{program} is a downward-closed set of finite partial strings with respect to $\sqsubseteq$; equivalently $\cX \subseteq \sP_f$ is a program whenever $\downarrow_\sqsubseteq \cX = \cX$ where $\downarrow_\sqsubseteq \cX \deq \set{y \in \sP_f \alt \exists x \in \cX \colon y \sqsubseteq x}$. Denote with $\bbP$ the family of all programs.
\end{definition}

Since we only consider systems that terminate, each partial string $x$ in a program $\cX$ is finite. We reemphasize that the downward closure of such a set $\cX$ can be thought of as an over-approximation of all possible happens-before relations in a concurrent system whose instructions are ordered according to the partial strings in $\cX$. Later on (\autoref{section:SC-relaxed}) we make the downward closure of partial strings more precise to model a certain kind of relaxed sequential consistency.

\begin{example}
\label{example:program}
Recall that $N$-shaped partial strings cannot be constructed as $x ; y$ or $x \parallel y$ under any labelling~\cite{P1986}. Yet, by downward-closure of programs, such partial strings are included in the over-approximation of all the happens-before relations exhibited by a concurrent system. In particular, according to Example~\ref{example:sqsubseteq-partial-string}, the downward-closure of the set containing the partial string in Fig.~\ref{fig:partial-string-example} includes (among many others) the $N$-shaped partial string shown on the right in Fig.~\ref{fig:sqsubseteq-partial-string}. In fact, we shall see in~\autoref{section:SC-relaxed} that this particular $N$-shaped partial string corresponds to a data race in the concurrent system shown in Fig.~\ref{fig:intro-example-racy}.
\end{example}

It is standard~\cite{G1988,HA2014} to define $0 \deq \emptyset$ and $1 \deq \set{\bot}$ where $\bot$ is the (unique) empty partial string. Clearly $0$ and $1$ form programs in the sense of Definition~\ref{def:program}. For the next theorem, we lift the two partial string operators (Definition~\ref{def:partial-string-composition}) to programs in the standard way:

\begin{definition}[Bow tie]
\label{def:program-operator}
Given two partial strings $x$ and $y$, denote with $x \Join y$ either concurrent or sequential composition of $x$ and $y$. For all programs $\cX, \cY$ in $\bbP$ and partial string operators $\Join$, $\cX \Join \cY \deq\ \downarrow_\sqsubseteq \set{x \Join y \alt x \in \cX \text{ and } y \in \cY}$ where $\cX \parallel \cY$ and $\cX ; \cY$ are called \defn{concurrent} and \defn{sequential program composition}, respectively.
\end{definition}

By denoting programs as sets of partial strings, we can now define Kleene star operators $(-)^\parallel$ and $(-)^;$ for iterative concurrent and sequential program composition, respectively, as least fixed points ($\mu$) using set union ($\cup$) as the binary join operator that we interpret as the nondeterministic choice of two programs. We remark that this is fundamentally different from the pomsets recursion operators in ultra-metric spaces~\cite{BW1990}. The next theorem could be then summarized as saying that the resulting structure of programs, written $\mfrS$, is a partial order model of an algebraic concurrency semantics that satisfies the CKA laws~\cite{HMSW2011}. Since CKA is an exemplar of the universal laws of programming~\cite{HvS2014}, we base the rest of this paper on our partial order model of CKA.

\begin{theorem}
\label{theorem:program-algebra}
The structure $\mfrS = \tuple{\bbP, \subseteq, \cup, 0, 1, ;, \parallel}$ is a complete lattice, ordered by subset inclusion (i.e. $\cX \subseteq \cY$ exactly if $\cX \cup \cY = \cY$), such that $\parallel$ and $;$ form unital quantales over $\cup$ where $\mfrS$ satisfies the following:
\begin{align*}
&(\cU \parallel \cV) ; (\cX \parallel \cY) \subseteq (\cU ; \cX) \parallel (\cV ; \cY) &\quad  \cX \cup (\cY \cup \cZ) = (\cX \cup \cY) \cup \cZ  \\
&\cX \cup \cX = \cX &\quad  \cX \cup 0 = 0 \cup \cX = \cX                                                                             \\
&\cX \cup \cY = \cY \cup \cX &\quad \cX \parallel \cY = \cY \parallel \cX                                                          \\
&\cX \parallel 1 = 1 \parallel \cX = \cX &\quad \cX ; 1 = 1 ; \cX = \cX                                                      \\
&\cX \parallel 0 = 0 \parallel \cX = 0   &\quad \cX ; 0 = 0 ; \cX = 0                                                        \\
&\cX \parallel (\cY \cup \cZ) = (\cX \parallel \cY) \cup (\cX \parallel \cZ) &\quad  \cX ; (\cY \cup \cZ) = (\cX ; \cY) \cup (\cX ; \cZ) \\
&(\cX \cup \cY) \parallel \cZ = (\cX \parallel \cZ) \cup (\cY \parallel \cZ) &\quad (\cX \cup \cY) ; \cZ = (\cX ; \cZ) \cup (\cY ; \cZ)  \\
&\cX \parallel (\cY \parallel \cZ) = (\cX \parallel \cY) \parallel \cZ &\quad \cX ; (\cY ; \cZ) = (\cX ; \cY) ; \cZ          \\
&\cP^{\parallel} = \mu \cX .1 \cup (\cP \parallel \cX) &\quad \cP^{;} = \mu \cX .1 \cup (\cP ; \cX).                                   
\end{align*}
\end{theorem}
\begin{proof}
The details are in the accompanying technical report of this paper~\cite{HA2014}.
\end{proof}

By Theorem~\ref{theorem:program-algebra}, it makes sense to call $1$ in structure $\mfrS$ the \defn{$\Join$-identity program} where $\Join$ is a placeholder for either $;$ or $\parallel$. In the sequel, we call the binary relation $\subseteq$ on $\bbP$ the \defn{program refinement relation}.

\section{Least fixed point reduction}
\vspace{-0.3em}
\label{label:least-fixed-point-reduction}

This section is about the least fixed point operators $(-)^;$ and $(-)^\parallel$. Henceforth, we shall denote these by $(-)^{\Join}$. We show that under a certain finiteness condition (Definition~\ref{def:program-elementary}) the program refinement problem $\cX^{\Join} \subseteq \cY^{\Join}$ can be reduced to a bounded number of program refinement problems without least fixed points (Theorem~\ref{theorem:program-reduction}). To prove this, we start by inductively defining the notion of iteratively composing a program with itself under $\Join$.

\begin{definition}[$n$-iterated-$\Join$-program-composition]
\label{def:program-n-iterated-composition}
Let $\nats_0 \deq \nats \cup \set{0}$ be the set of \defn{non-negative integers}. For all programs $\cP$ in $\bbP$ and non-negative integers $n$ in $\nats_0$, $\cP^{0 \cdot \Join} \deq 1 = \set{\bot}$ is the $\Join$-identity program and $\cP^{(n+1) \cdot \Join} \deq \cP \Join \cP^{n \cdot \Join}$.
\end{definition}

Clearly $(-)^{\Join}$ is the limit of its approximations in the following sense:

\begin{proposition}
\label{proposition:program-least-fixed-point-as-n-iterated-composition}
For every program $\cP$ in $\bbP$, $\cP^{\Join} = \bigcup_{n \ge 0} \cP^{n \cdot \Join}$.
\end{proposition}

\begin{definition}[Elementary program]
\label{def:program-elementary}
A program $\cP$ in $\bbP$ is called \defn{elementary} if $\cP$ is the downward-closed set with respect to $\sqsubseteq$ of some finite and nonempty set $\cQ$ of finite partial strings, i.e. $\cP = \downarrow_\sqsubseteq \cQ$. The set of elementary programs is denoted by $\bbP_\ell$.
\end{definition}

An elementary program therefore could be seen as a machine-representable program generated from a finite and nonempty set of finite partial strings. This finiteness restriction makes the notion of elementary programs a suitable candidate for the study of decision procedures. To make this precise, we define the following unary partial string operator:

\begin{definition}[$n$-repeated-$\Join$ partial string operator]
\label{def:partial-string-dot-coproduct}
For every non-negative integer $n$ in $\nats_0$, $x^{0 \cdot \Join} \deq \bot$ is the empty partial string and $x^{(n+1) \cdot \Join} \deq x \Join x^{n \cdot \Join}$.
\end{definition}

Intuitively, $p^{n \cdot \Join}$ is a partial string that consists of $n$ copies of a partial string $p$, each combined by the partial string operator $\Join$. This is formalized as follows:

\begin{proposition}
Let $n \in \nats_0$ be a non-negative integer. Define $[0] \deq \emptyset$ and $[n + 1] \deq \set{1, \ldots, n + 1}$. For every partial string $x$, $x^{n \cdot \Join}$ is isomorphic to $y = \tuple{E_y, \alpha_y, \preceq_y}$ where $E_y \deq E_x \times [n]$ such that, for all $e, e' \in E_x$ and $i, i' \in [n]$, the following holds:
\begin{itemize}
\item if `$\Join$' is `$\parallel$', then $\pair{e}{i} \preceq_y \pair{e'}{i'}$ exactly if $i = i'$ and $e \preceq_x e'$,
\item if `$\Join$' is `$;$', then $\pair{e}{i} \preceq_y \pair{e'}{i'}$ exactly if $i < i'$ or ($i = i'$ and $e \preceq_x e'$),
\item $\alpha_y(\pair{e}{i}) = \alpha_x(e)$.
\end{itemize}
\end{proposition}

\begin{definition}[Partial string size]
\label{def:partial-string-size}
The \defn{size} of a finite partial string $p$, denoted by $\abs{p}$, is the cardinality of its event set $E_p$.
\end{definition}

For example, the partial string in Fig.~\ref{fig:partial-string-example} has size four. It is obvious that the size of finite partial strings is non-decreasing under the $n$-repeated-$\Join$ partial string operator from Definition~\ref{def:partial-string-dot-coproduct} whenever $0 < n$. This simple fact is important for the next step towards our least fixed point reduction result in Theorem~\ref{theorem:program-reduction}:

\begin{proposition}[Elementary least fixed point pre-reduction]
\label{proposition:program-forward-finite-reduction}
For all elementary programs $\cX$ and $\cY$ in $\bbP_\ell$, if the $\Join$-identity program $1$ is not in $\cY$ and $\cX \subseteq \cY^{\Join}$, then $\cX \subseteq \bigcup_{n \ge k \ge 0} \cY^{k \cdot \Join}$ where $n = \floor[\Big]{\frac{\ell_\cX}{\ell_\cY}}$ such that $\ell_\cX \deq \mathsf{max}\left\{\abs{x} \alt x \in \cX\right\}$ and $\ell_\cY \deq \mathsf{min}\left\{\abs{y} \alt y \in \cY\right\}$ is the size of the largest and smallest partial strings in $\cX$ and $\cY$, respectively.
\end{proposition}
\begin{proof}
Assume $\cX \subseteq \cY^{\Join}$. Let $x \in \sP_f$ be a finite partial string. We can assume $x \in \cX$ because $\cX \not= 0$. By assumption, $x \in \cY^{\Join}$. By Proposition~\ref{proposition:program-least-fixed-point-as-n-iterated-composition}, there exists $k \in \nats_0$ such that $x \in \cY^{k \cdot \Join}$. Fix $k$ to be the smallest such non-negative integer. Show $k \leq \floor[\Big]{\frac{\ell_\cX}{\ell_\cY}}$ (the fraction is well-defined because $\cX$ and $\cY$ are nonempty and $1 \not\in \cY$). By downward closure and definition of $\sqsubseteq$ in terms of a one-to-one correspondence, it suffices to consider that $x$ is one of a (not necessarily unique) longest partial strings in $\cX$, i.e. $\abs{x'} \leq \abs{x}$ for all $x' \in \cX$; equivalently, $\abs{x} = \ell_\cX$. If $\abs{x} = 0$, set $k = 0$, satisfying $1 = \cX \subseteq \cY^{k \cdot \Join} = 1$ and $k \leq n = 0$ as required. Otherwise, since the size of partial strings in a program can never decrease under the $k$-iterated program composition operator $\Join$ when $0 < k$, it suffices to consider the case $x \sqsubseteq y^{k \cdot \Join}$ for some shortest partial string $y$ in $\cY$. Since $E_{y^{k \cdot \Join}}$ is the Cartesian product of $E_y$ and $[k]$, it follows $\abs{x} = k \cdot \abs{y}$. Since $\abs{x} \leq \ell_\cX$ and $\ell_\cY \leq \abs{y}$, $k \leq \floor[\big]{\frac{\ell_\cX}{\ell_\cY}}$. By definition $n = \floor[\Big]{\frac{\ell_\cX}{\ell_\cY}}$, proving $x \in \bigcup_{n \ge k \ge 0} \cY^{k \cdot \Join}$. \qed
\end{proof}

Equivalently, if there exists a partial string $x$ in $\cX$ such that $x \not\in \cY^{k \cdot \Join}$ for all non-negative integers $k$ between zero and $\floor[\Big]{\frac{\ell_\cX}{\ell_\cY}}$, then $\cX \not\subseteq \cY^{\Join}$. Since we are interested in decision procedures for program refinement checking, we need to show that the converse of Proposition~\ref{proposition:program-forward-finite-reduction} also holds. Towards this end, we prove the following left $(-)^{\Join}$ elimination rule:

\begin{proposition}
\label{proposition:program-eliminate-left-lfp}
For every program $\cX$ and  $\cY$ in $\bbP$, $\cX^{\Join} \subseteq \cY^{\Join}$ exactly if $\cX \subseteq \cY^{\Join}$.
\end{proposition}
\begin{proof}
Assume $\cX^{\Join} \subseteq \cY^{\Join}$. By Proposition~\ref{proposition:program-least-fixed-point-as-n-iterated-composition}, $\cX \subseteq \cX^{\Join}$. By transitivity of $\subseteq$ in $\bbP$, $\cX \subseteq \cY^{\Join}$. Conversely, assume $\cX \subseteq \cY^{\Join}$. Let $i, j \in \nats_0$. By induction on $i$, $\cX^{i \cdot \Join} \Join \cX^{j \cdot \Join} = \cX^{(i+j) \cdot \Join}$. Thus, by Proposition~\ref{proposition:program-least-fixed-point-as-n-iterated-composition} and distributivity of $\Join$ over least upper bounds in $\bbP$, $\cX^{\Join} \Join \cX^{\Join} = \cX^{\Join}$, i.e. $(-)^{\Join}$ is idempotent. This, in turn, implies that $(-)^{\Join}$ is a closure operator. Therefore, by monotonicity, $\cX^{\Join} \subseteq \left(\cY^{\Join}\right)^{\Join} = \cY^{\Join}$, proving that $\cX^{\Join} \subseteq \cY^{\Join}$ is equivalent to $\cX \subseteq \cY^{\Join}$. \qed
\end{proof}

\begin{theorem}[Elementary least fixed point reduction]
\label{theorem:program-reduction}
For all elementary programs $\cX$ and $\cY$ in $\bbP_\ell$, if the $\Join$-identity program $1$ is not in $\cY$, then $\cX^{\Join} \subseteq \cY^{\Join}$ is equivalent to $\cX \subseteq \bigcup_{n \ge k \ge 0} \cY^{k \cdot \Join}$ where $n = \floor[\Big]{\frac{\ell_\cX}{\ell_\cY}}$ such that $\ell_\cX \deq \mathsf{max}\left\{\abs{x} \alt x \in \cX\right\}$ and $\ell_\cY \deq \mathsf{min}\left\{\abs{y} \alt y \in \cY\right\}$ is the size of the largest and smallest partial strings in $\cX$ and $\cY$, respectively.
\end{theorem}
\begin{proof}
By Proposition~\ref{proposition:program-eliminate-left-lfp}, it remains to show that $\cX \subseteq \cY^{\Join}$ is equivalent to $\cX \subseteq \bigcup_{n \ge k \ge 0} \cY^{k \cdot \Join}$ where $n = \floor[\Big]{\frac{\ell_\cX}{\ell_\cY}}$. The forward and backward implication follow from Proposition~\ref{proposition:program-forward-finite-reduction}~and~\ref{proposition:program-least-fixed-point-as-n-iterated-composition}, respectively. \qed
\end{proof}

From Theorem~\ref{theorem:program-reduction} follows immediately that $\cX^{\Join} \subseteq \cY^{\Join}$ is decidable for all elementary programs $\cX$ and $\cY$ in $\bbP_\ell$ because there exists an algorithm that could iteratively make $O\left(\abs{\cX} \times \abs{\cY}^n\right)$ calls to another decision procedure to check whether $x \sqsubseteq y$ for all $x \in \cX$ and $y \in \cY^{k \cdot \Join}$ where $n \geq k \geq 0$. However, by Proposition~\ref{proposition:psr-NP-complete}, each iteration in such an algorithm would have to solve an NP-complete subproblem. But this high complexity is expected since the PLC problem is $\Pi_2^p$-complete~\cite{FKL1993}.

\begin{corollary}
For all elementary programs $\cX$ and $\cY$ in $\bbP$, if $\abs{x} = \abs{y}$ for all $x \in \cX$ and $y \in \cY$, then $\cX^{\Join} \subseteq \cY^{\Join}$ is equivalent to $\cX \subseteq \cY$.
\end{corollary}

We next move on to enriching our model of computation to accommodate a certain kind of relaxed sequential consistency.

\section{Relaxed sequential consistency}
\vspace{-0.3em}
\label{section:SC-relaxed}

For efficiency reasons, all modern computer architectures implement some form of weak memory model rather than sequential consistency~\cite{L1979}. A defining characteristic of weak memory architectures is that they violate interleaving semantics unless specific instructions are used to restore sequential consistency. This section fixes a particular interpretation of weak memory and studies the mathematical properties of the resulting partial order semantics. For this, we separate memory accesses into synchronizing and non-synchronizing ones, akin to~\cite{GLLGGH1990}. A synchronized store is called a \emph{release}, whereas a synchronized load is called an \emph{acquire}. The intuition behind release/acquire is that prior writes made to other memory locations by the thread executing the release become visible in the thread that performs the corresponding acquire. Crucially, the particular form of release/acquire semantics that we formalize here is shown to be equivalent to the conjunction of three weak memory axioms (Theorem~\ref{theorem:SC-relaxed-equivalence}), namely `write coherence', `from-read' and `global read-from'~\cite{AMSS2012}. Subsequently, we look at one important ramification of this equivalence on \emph{bounded model checking} (BMC) techniques for finding concurrency-related bugs (Theorem~\ref{theorem:smaller-fr}).

We start by defining the alphabet that we use for identifying events that denote synchronizing and non-synchronizing memory accesses.

\newcommand{\memorylocation}{\langle\textit{ADDRESS}\rangle}
\newcommand{\register}{\langle\textit{REG}\rangle}
\newcommand{\loadmemoryorder}{\langle\textit{LOAD}\rangle}
\newcommand{\storememoryorder}{\langle\textit{STORE}\rangle}
\newcommand{\bit}{\langle\textit{BIT}\rangle}

\begin{definition}[Memory access alphabet]
\label{def:memory-access-alphabet}
Define $\loadmemoryorder \deq \set{\mathsf{none}, \mathsf{acquire}}$, $\storememoryorder \deq \set{\mathsf{none}, \mathsf{release}}$ and $\bit \deq \set{0,1}$. Let $\memorylocation$ and $\register$ be disjoint sets of \defn{memory locations} and \defn{registers}, respectively. Let $\mathit{load\_tag} \in \loadmemoryorder$ and $\mathit{store\_tag} \in \storememoryorder$. Define the set of \defn{load} and \defn{store} labels, respectively:
\label{def:alphabet}
\begin{align*}
\Gamma_{\mathsf{load},\,\mathit{load\_tag}} &\deq \set{\mathit{load\_tag}} \times \register \times \memorylocation \\
\Gamma_{\mathsf{store},\,\mathit{store\_tag}} &\deq \set{\mathit{store\_tag}} \times \memorylocation \times \bit
\end{align*}
Let $\Gamma \deq \Gamma_{\mathsf{load}, \mathsf{none}} \cup \Gamma_{\mathsf{load}, \mathsf{acquire}} \cup \Gamma_{\mathsf{store}, \mathsf{none}} \cup \Gamma_{\mathsf{store}, \mathsf{release}}$ be the \defn{memory access alphabet}. Given $r \in \register$, $a \in \memorylocation$ and $b \in \bit$, we write $`r\, \texttt{:=}\, [a]_\mathit{load\_tag}\textrm'$ for the label $\tuple{\mathit{load\_tag}, r, a}$ in $\Gamma_{\mathsf{load},\,\mathit{load\_tag}}$; similarly, $`[a]_\mathit{store\_tag}\,\texttt{:=}\,b\textrm'$ is shorthand for the label $\tuple{\mathit{store\_tag}, a, b}$ in $\Gamma_{\mathsf{store},\,\mathit{store\_tag}}$.

Let $x$ be a partial string and $e$ be an event in $E_x$. Then $e$ is called a \defn{load} or \defn{store} if its label, $\alpha_x(e)$, is in $\Gamma_{\mathsf{load},\,\mathit{load\_tag}}$ or $\Gamma_{\mathsf{store},\,\mathit{store\_tag}}$, respectively. A load or store event $e$ is a \defn{non-synchronizing memory access} if $\alpha_x(e) \in \Gamma_\mathsf{none} \deq \Gamma_{\mathsf{load}, \mathsf{none}} \cup \Gamma_{\mathsf{store}, \mathsf{none}}$; otherwise, it is a \defn{synchronizing memory access}. Let $a \in \memorylocation$ be a memory location. An \defn{acquire on $a$} is an event $e$ such that $\alpha_x(e) = `r\, \texttt{:=}\, [a]_\mathsf{acquire}\textrm'$ for some $r \in \register$. Similarly, a \defn{release on $a$} is an event $e$ labelled by $`[a]_\mathsf{release}\,\texttt{:=}\,b\textrm'$ for some $b \in \bit$. A \defn{release} and \defn{acquire} is a release and acquire on some memory location, respectively.
\end{definition}

\begin{SCfigure}[100][t]
\begin{tabular}{@{}l@{\hspace{2mm}} || @{\hspace{2mm}}l@{}}
\multicolumn{1}{c}{Thread $\texttt{T}_1$}    & \multicolumn{1}{c}{Thread $\texttt{T}_2$}       \\
\midrule
  $r_0$\,\texttt{:=}\,$[b]_\mathsf{acquire}$ & $[a]_\mathsf{none}$\,\texttt{:=}\,\texttt{1}    \\
  $r_1$\,\texttt{:=}\,$[a]_\mathsf{none}$    & $[b]_\mathsf{release}$\,\texttt{:=}\,\texttt{1}
\end{tabular}
\caption{A concurrent system $\texttt{T}_1\,\parallel\,\texttt{T}_2$ consisting of two threads. The memory accesses on memory locations $b$ are synchronized, whereas those on $a$ are not.}
\label{fig:intro-example-racy}
\end{SCfigure}

\begin{example}
\label{example:memory-access}
Fig.~\ref{fig:intro-example-racy} shows the syntax of a program that consists of two threads $\texttt{T}_1$ and $\texttt{T}_2$. This concurrent system can be directly modelled by the partial string shown in Fig.~\ref{fig:partial-string-example} where memory location $b$ is accessed through acquire and release, whereas memory location $a$ is accessed through non-synchronizing loads and stores (shortly, we shall see that this leads to a data race).
\end{example}

Given Definition~\ref{def:memory-access-alphabet}, we are now ready to refine our earlier conservative over-approximation of the happens-before relations (Definition~\ref{def:program}) to get a particular form of release/acquire semantics. For this, we restrict the downward closure of programs $\cX$ in $\bbP$, in the sense of Definition~\ref{def:program}, by requiring all partial strings in $\cX$ to satisfy the following partial ordering constraints:

\begin{definition}[SC-relaxed program]
\label{def:SC-relaxed-program}
A program $\cX$ is called \defn{SC-relaxed} if, for all $a \in \memorylocation$ and partial string $x$ in $\cX$, the set of release events on $a$ is totally ordered by $\preceq_x$ and, for every acquire $l \in E_x$ and release $s \in E_x$ on $a$, $l \preceq_x s$ or $s \preceq_x l$.
\end{definition}

Henceforth, we denote loads and stores by $l, l'$ and $s, s'$, respectively. If $s$ and $s'$ are release events that modify the same memory location, either $s$ happens-before $s'$, or vice versa.  If $l$ is an acquire and $s$ is a release on the same memory location, either $l$ happens-before $s$ or $s$ happens-before $l$. Importantly, however, two acquire events $l$ and $l'$ on the same memory location may still happen concurrently in the sense that neither $l$ happens-before $l'$ nor $l'$ happens-before $l$, in the same way non-synchronizing memory accesses are generally unordered.

\begin{example}
\label{example:memory-access-with-data-race}
Example~\ref{example:program}~and~\ref{example:memory-access} illustrate the SC-relaxed semantics of the concurrent system in Fig.~\ref{fig:intro-example-racy}. In particular, the $N$-shaped partial string in Fig.~\ref{fig:sqsubseteq-partial-string} corresponds to a data race in $\texttt{T}_1\,\parallel\,\texttt{T}_2$ because the non-synchronizing memory accesses on memory location $a$ happen concurrently. To see this, it may help to consider the interleaving $r_0\, \texttt{:=}\, [b]_\mathsf{acquire} ; [a]_\mathsf{none}\,\texttt{:=}\,1 ; r_1\, \texttt{:=}\, [a]_\mathsf{none} ; [b]_\mathsf{release}\,\texttt{:=}\,1$ where both memory accesses on location $a$ are unordered through the happens-before relation because there is no release instruction separating $[a]_\mathsf{none}\,\texttt{:=}\,1$ from $r_1\, \texttt{:=}\, [a]_\mathsf{none}$. One way of fixing this data race is by changing thread $\texttt{T}_1$ to $\mathbf{if}\ [b]_\mathsf{acquire} = 1\ \mathbf{then}\ r_1\,\texttt{:=}\,[a]_\mathsf{none}$. Since CKA supports non-deterministic choice with the $\cup$ binary operator (recall Theorem~\ref{theorem:program-algebra}), it would not be difficult to give semantics to such conditional checks, particularly if we introduce `assume' labels into the alphabet in Definition~\ref{def:memory-access-alphabet}.
\end{example}

We ultimately want to show that the conjunction of three existing weak memory axioms as studied in~\cite{AMSS2012} fully characterizes our particular interpretation of relaxed sequential consistency, thereby paving the way for Theorem~\ref{theorem:smaller-fr}. For this, we recall the following memory axioms which can be thought of as relations on loads and stores on the same memory location:

\begin{definition}[Memory axioms]
\label{def:memory-axioms}
Let $x$ be a partial string in $\sP_f$. The \defn{read-from} function, denoted by $\mathsf{rf} \colon E_x \to E_x$, is defined to map every load to a store on the same memory location. A load $l$ \defn{synchronizes-with} a store $s$ if $\mathsf{rf}(l) = s$ implies $s \preceq_x l$. \defn{Write-coherence} means that all stores $s, s'$ on the same memory location are totally ordered by $\preceq_x$. The \defn{from-read axiom} holds whenever, for all loads $l$ and stores $s, s'$ on the same memory location, if $\mathsf{rf}(l) = s$ and $s \prec_x s'$, then $l \preceq_x s'$.
\end{definition}

By definition, the read-from function is total on all loads. The synchronizes-with axiom says that if a load reads-from a store (necessarily on the same memory location), then the store happens-before the load. This is also known as the global read-from axiom~\cite{AMSS2012}. Write-coherence, in turn, ensures that all stores on the same memory location are totally ordered. This corresponds to the fact that ``all writes to the same location are serialized in some order and are performed in that order with respect to any processor''~\cite{GLLGGH1990}. Note that this is different from the modification order (`mo') on atomics in C++14~\cite{CPP14} because `mo' is generally not a subset of the happens-before relation. The from-read axiom~\cite{AMSS2012} requires that, for all loads $l$ and two different stores $s, s'$ on the same location, if $l$ reads-from $s$ and $s$ happens-before $s'$, then $l$ happens-before $s'$.
We start by deriving from these three memory axioms the notion of SC-relaxed programs.

\begin{proposition}[SC-relaxed consistency]
\label{proposition:SC-relaxed-consistency}
For all $\cX$ in $\bbP$, if, for each partial string $x$ in $\cX$, the synchronizes-with, write-coherence and from-read axioms hold on all release and acquire events in $E_x$ on the same memory location, then $\cX$ is an SC-relaxed program.
\end{proposition}
\begin{proof}
Let $a \in \memorylocation$ be a memory location, $l$ be an acquire on $a$ and $s'$ be a release on $a$. By write-coherence on release/acquire events, it remains to show $l \preceq_x s'$ or $s' \preceq_x l$. Since the read-from function is total, $\mathsf{rf}(l) = s$ for some release $s$ on $a$. By the synchronizes-with axiom, $s \preceq_x l$. We therefore assume $s \not= s'$. By write-coherence, $s \prec_x s'$ or $s' \prec_x s$. The former implies $l \preceq_x s'$ by the from-read axiom, whereas the latter implies $s' \preceq_x l$ by transitivity. This proves, by case analysis, that $\cX$ is an SC-relaxed program. \qed
\end{proof}

We need to prove some form of converse of the previous implication in order to characterize SC-relaxed semantics in terms of the three aforementioned weak memory axioms. For this purpose, we define the following:

\begin{definition}[Read consistency]
\label{def:read-consistency}
Let $a \in \memorylocation$ be a memory location and $x$ be a finite partial string in $\sP_f$. For all loads $l \in E_x$ on $a$, define the following set of store events: $\cH_x(l) \deq \set{s \in E_x \alt s \preceq_x l \text{ and } s \text{ is a store on } a}$. The read-from function $\mathsf{rf}$ is said to satisfy \defn{weak read consistency} whenever, for all loads $l \in E_x$ and stores $s \in E_x$ on memory location $a$, the least upper bound $\bigvee \cH_x(l)$ exists, and $\mathsf{rf}(l) = s$ implies $\bigvee \cH_x(l) \preceq_x s$; \defn{strong read consistency} implies $\mathsf{rf}(l) = s = \bigvee \cH_x(l)$.
\end{definition}

By the next proposition, a natural sufficient condition for the existence of the least upper bound $\bigvee \cH_x(l)$ is the finiteness of the partial strings in $\sP_f$ and the total ordering of all stores on the same memory location from which the load $l$ reads, i.e. write coherence. This could be generalized to well-ordered sets.

\begin{proposition}[Weak read consistency existence]
\label{proposition:weak-read-consistency-existence}
For all partial strings $x$ in $\sP_f$, write coherence on memory location $a$ implies that $\bigvee \cH_x(l)$ exists for all loads $l$ on $a$.
\end{proposition}

We remark that $\bigvee \cH_x(l) = \bot$ if $\cH_x(l) = \emptyset$; alternatively, to avoid that $\cH_x(l)$ is empty, we could require that programs are always constructed such that their partial strings have minimal store events that initialize all memory locations.

\begin{proposition}[Weak read consistency equivalence]
\label{proposition:read-consistency-characterization}
Write coherence implies that weak read consistency is equivalent to the following: for all loads $l$ and stores $s, s'$ on memory location $a \in \memorylocation$, if $\mathsf{rf}(l) = s$ and $s' \preceq_x l$, then $s' \preceq_x s$.
\end{proposition}
\begin{proof}
By write coherence, $\bigvee \cH_x(l)$ exists, and $s' \preceq_x \bigvee \cH_x(l)$ because $s' \in \cH_x(l)$ by assumption $s' \preceq_x l$ and Definition~\ref{def:read-consistency}. By assumption of weak read consistency, $\bigvee \cH_x(l) \preceq_x s$. From transitivity follows $s' \preceq_x s$.

Conversely, assume $\mathsf{rf}(l) = s$. Let $s'$ be a store on $a$ such that $s' \in \cH_x(l)$. Thus, by hypothesis, $s' \preceq_x s$. Since $s'$ is arbitrary, $s$ is an upper bound. Since the least upper bound is well-defined by write coherence, $\bigvee \cH_x(l) \preceq_x s$. \qed
\end{proof}

Weak read consistency therefore says that if a load $l$ reads from a store $s$ and another store $s'$ on the same memory location happens before $l$, then $s'$ happens before $s$. This implies the next proposition. 

\begin{proposition}[From-read equivalence]
\label{proposition:read-consistency-equivalence}
For all SC-relaxed programs in $\bbP$, weak read consistency with respect to release/acquire events is equivalent to the from-read axiom with respect to release/acquire events.
\end{proposition}

We can characterize strong read consistency as follows:

\begin{proposition}[Strong read consistency equivalence]
\label{proposition:strong-read-consistency-equivalence}
Strong read consistency is equivalent to weak read consistency and the synchronizes-with axiom.
\end{proposition}
\begin{proof}
Let $x$ be a partial string in $\sP_f$. Let $l$ be a load and $s$ be a store on the same memory location. The forward implication is immediate from $\bigvee \cH_x(l) \preceq_x l$.

Conversely, assume $\mathsf{rf}(l) = s$. By synchronizes-with, $s \preceq_x l$, whence $s \in \cH_x(l)$. By definition of least upper bound, $s \preceq_x \bigvee \cH_x(l)$. Since $s \succeq_x \bigvee \cH_x(l)$, by hypothesis, and $\preceq_x$ is antisymmetric, we conclude $s = \bigvee \cH_x(l)$. \qed
\end{proof}

\begin{theorem}[SC-relaxed equivalence]
\label{theorem:SC-relaxed-equivalence}
For every program $\cX$ in $\bbP$, $\cX$ is SC-relaxed where, for all partial strings $x$ in $\cX$ and acquire events $l$ in $E_x$, $\mathsf{rf}(l) = \bigvee \cH_x(l)$, if and only if the synchronizes-with, write-coherence and from-read axioms hold for all $x$ in $\cX$ with respect to all release/acquire events in $E_x$ on the same memory location.
\end{theorem}
\begin{proof}
Assume $\cX$ is an SC-relaxed program according to Definition~\ref{def:SC-relaxed-program}. Let $x$ be a partial string in $\cX$ and $l$ be an acquire in the set of events $E_x$. By Proposition~\ref{proposition:weak-read-consistency-existence}, $\bigvee \cH_x(l)$ exists. Assume $\mathsf{rf}(l) = \bigvee \cH_x(l)$. Since $l$ is arbitrary, this is equivalent to assuming strong read consistency. Since release events are totally ordered in $\preceq_x$, by assumption, it remains to show that the synchronizes-with and from-read axioms hold. This follows from Proposition~\ref{proposition:strong-read-consistency-equivalence}~and~\ref{proposition:read-consistency-equivalence}, respectively.

Conversely, assume the three weak memory axioms hold on $x$ with respect to all release/acquire events in $E_x$ on the same memory location. By Proposition~\ref{proposition:SC-relaxed-consistency}, $\cX$ is an SC-relaxed program. Therefore, by Proposition~\ref{proposition:read-consistency-equivalence}~and~\ref{proposition:strong-read-consistency-equivalence}, $\mathsf{rf}(l) = \bigvee \cH_x(l)$, proving the equivalence. \qed
\end{proof}

While the state-of-the-art weak memory encoding is cubic in size~\cite{AKT2013}, the previous theorem has as immediate consequence that there exists an asymptotically smaller weak memory encoding with only a quadratic number of partial order constraints.

\begin{theorem}[Quadratic-size weak memory encoding]
\label{theorem:smaller-fr}
There exists a quantifier-free first-order logic formula that has a quadratic number of partial order constraints and is equisatisfiable to the cubic-size encoding given in~\cite{AKT2013}.
\end{theorem}

\begin{proof}
Instead of instantiating the three universally quantified events in the from-read axiom, symbolically encode the least upper bound of weak read consistency. This can be accomplished with a new symbolic variable for every acquire event. It is easy to see that this reduces the cubic number of partial order constraints to a quadratic number. \qed
\end{proof}

In short, the asymptotic reduction in the number of partial order constraints is due to a new symbolic encoding for how values are being overwritten in memory: the current cubic-size formula~\cite{AKT2013} encodes the from-read axiom (Definition~\ref{def:memory-axioms}), whereas the proposed quadratic-size formula encodes a certain least upper bound (Definition~\ref{def:read-consistency}). We reemphasize that this formulation is in terms of release/acquire events rather than machine-specific accesses as in~\cite{AKT2013}. The construction of the quadratic-size encoding, therefore, is generally only applicable if we can translate the machine-specific reads and writes in a shared memory program to acquire and release events, respectively. This may require the program to be data race free, as illustrated in Example~\ref{example:memory-access-with-data-race}.

Furthermore, as mentioned in the introduction of this section, the primary application of Theorem~\ref{theorem:smaller-fr} is in the context of BMC. Recall that BMC assumes that all loops in the shared memory program under scrutiny have been unrolled (the same restriction as in~\cite{AKT2013}). This makes it possible to symbolically encode branch conditions, thereby alleviating the need to explicitly enumerate each finite partial string in an elementary program.

\section{Concluding remarks}
\vspace{-0.3em}
\label{section:concl}

This paper has studied a partial order model of computation that satisfies the axioms of a unifying algebraic concurrency semantics by Hoare et al. By further restricting the partial string semantics, we obtained a relaxed sequential consistency semantics which was shown to be equivalent to the conjunction of three weak memory axioms by Alglave et al. This allowed us to prove the existence of an equisatisfiable but asymptotically smaller weak memory encoding that has only a quadratic number of partial order constraints compared to the state-of-the-art cubic-size encoding. In upcoming work, we will experimentally compare both encodings in the context of bounded model checking using SMT solvers. As future theoretical work, it would be interesting to study the relationship between categorical models of partial string theory and event structures.

\paragraph{Acknowledgements.} We would like to thank Tony Hoare and Stephan van Staden for their valuable comments on an early draft of this paper, and we thank Jade Alglave, C\'{e}sar Rodr\'{i}guez, Michael Tautschnig, Peter Schrammel, Marcelo Sousa, Bj\"{o}rn Wachter and John Wickerson for invaluable discussions.

\bibliographystyle{splncs}
\bibliography{paper}

\end{document}

%% file: macros.tex
\newcommand{\pair}[2]{\mbox{$\langle #1,\; #2 \rangle$}}
\newcommand{\tuple}[1]{\mbox{$\langle #1 \rangle$}}

\newcommand{\set}[1]{\mbox{$\{ #1 \}$}}
\newcommand{\alt}{\mathrel{|}}

\newcommand{\deq}{\triangleq}

\newcommand{\nats}{\mathbb{N}}

\newcommand{\sP}{\mathsf{P}}

\newcommand{\cH}{\mathcal{H}}

\newcommand{\cP}{\mathcal{P}}
\newcommand{\cQ}{\mathcal{Q}}

\newcommand{\cU}{\mathcal{U}}
\newcommand{\cV}{\mathcal{V}}

\newcommand{\cX}{\mathcal{X}}
\newcommand{\cY}{\mathcal{Y}}
\newcommand{\cZ}{\mathcal{Z}}

\newcommand{\mfrL}{\mathfrak{L}}

\newcommand{\mfrS}{\mathfrak{S}}

\newcommand{\bbP}{\mathbb{P}}

\newcommand{\defn}[1]{\textbf{#1}}

\DeclarePairedDelimiter\abs{\lvert}{\rvert}

\DeclarePairedDelimiter{\floor}{\lfloor}{\rfloor}